\documentclass[proceedings%
]{aofa}

\usepackage[utf8]{inputenc}
\usepackage{subfigure}

\usepackage{amsmath} 
\usepackage{amssymb}  

\newcommand{\poly}{\mathrm{poly}}
\newcommand{\defeq}{\stackrel{\mbox{\scriptsize{\normalfont\rmfamily def. }}}{=}}
\newcommand{\shortqed}{\hfill \mbox{$\blacksquare$} \smallskip}
\newcommand{\Order}{\mathrm{O}}
 
\newcommand{\Ni}{{\cal N}}
\newcommand{\No}{{\cal N}}
\newcommand{\N}{{\cal N}}

\newcommand{\dtv}{{\cal D}_{\rm tv}}

\newcommand{\CT}{T_{cover}}
\newcommand{\X}{X}
\newcommand{\nv}{X}
\newcommand{\M}{M}
\newcommand{\ti}{t'}
\newcommand{\tn}{k}
\newtheorem{theorem}{Theorem}[section]
\newtheorem{lemma}[theorem]{Lemma}
\newtheorem{corollary}[theorem]{Corollary}
\newtheorem{proposition}[theorem]{Proposition}

\author[Takeharu Shiraga]{Takeharu Shiraga\addressmark{1}\thanks{Supported by JSPS KAKENHI Grant Number 15J03840.}}

\title[The Cover Time of Deterministic Random Walks for General Transition Probabilities]{The Cover Time of Deterministic Random Walks for General Transition Probabilities}

\address{\addressmark{1}Graduate School of Information Science and Electrical Engineering, Kyushu University, Fukuoka, Japan\protect \\ 
 {\ttfamily takeharu.shiraga@inf.kyushu-u.ac.jp}}

\begin{document}
\maketitle
\begin{abstract}
\paragraph{Abstract.}
The {\em deterministic random walk} is a deterministic process
analogous to a random walk.
While there are some results on the cover time of the {\em rotor-router} model,
which is a deterministic random walk corresponding to a simple random walk,
nothing is known about the cover time of deterministic random walks
emulating general transition probabilities.
This paper is concerned with the {\em SRT-router} model with multiple tokens,
which is a deterministic process coping with general transition
probabilities possibly containing irrational numbers.
For the model, we give an upper bound of the cover time,
which is the first result on the cover time of deterministic random
walks for general transition probabilities.
Our upper bound also improves the existing bounds for the
rotor-router model in some cases.
%
\end{abstract}

\section{Introduction}
\label{sec:in}
\paragraph{Previous works for the cover time of random walks}%
A {\em random walk} is a fundamental stochastic process on a graph, in which a token successively transits to neighboring vertices chosen at random.
The expected cover time (this paper simply says {\em cover time}) of a random walk on a finite graph is the expected time until every vertex has been visited by the token.
The cover time is a fundamental measure of a random walk, and it has been well investigated.


Aleliunas et al.~\cite{AKLL79} showed that the cover time of a {\em simple random walk}, in which a neighboring vertex is chosen uniformly at random, 
is upper bounded by $2m(n-1)$ for any connected graph, where $m$ denotes the number of edges and $n$ denotes the number of vertices. 
Feige~\cite{F951,F952} showed that the cover time is lower bounded by $\bigl(1-o(1)\bigr)n\log n$ and upper bounded by $\bigl(1+o(1)\bigr)(4/27)n^3$ for any graph. 

Motivated by a faster cover time, the cover time by more than one token has also been investigated. 
Broder et al.~\cite{BKRU94} gave an upper bound of the cover time of $k$ independent parallel simple random walks ($k$-simple random walks) when tokens start from stationary distribution.
For an arbitrary initial configuration of tokens, Alon et al.~\cite{Alon11} showed that the cover time of $k$-simple random walks is upper bounded by $\bigl(({\rm e}+o(1))/k\bigr)t_{{\rm hit}}\log n$ for any graph if $k\leq \log n$, 
where ${\rm e}$ is Napier's constant and $t_{{\rm hit}}$ denotes the (maximum) {\em hitting time}.
Elsasser and Sauerwald~\cite{ES11} gave an better upper bound for large $\tn$ of $\Order\bigl(t^* +(t_{{\rm hit}}\log n)/\tn\bigr)$ for any graph if $\tn \leq n$, where $t^*$ is the {\em mixing time}.  

Ikeda et al.~\cite{IKY09} took another approach for speeding up, which uses {\em general} transition probabilities (beyond simple random walks). 
They invented {\em $\beta$-random walk}, consisting of irrational transition probabilities in general, and showed that the cover time is $\Order(n^2\log n)$. 
Nonaka et al.~\cite{NOSY10} showed that the cover time of a {\em Metropolis-walk}, which is based on the {\em Metropolis-Hastings algorithm}, is $\Order(n^2 \log n)$ for any graph. 

Little is known about the cover time by multiple tokens with general transition probabilities. 
Elsasser and Sauerwald~\cite{ES11} gave a general lower bound of $\Omega \bigl((n\log n)/\tn \bigr)$ for any transition probabilities and for any $n^\varepsilon \leq \tn \leq n$, where $0<\varepsilon <1$ is a constant. 
\paragraph{Previous works for the cover time of deterministic random walks}%
From the view point of the {\em deterministic} graph exploration, the {\em rotor-router model}, which is a deterministic process analogous to a simple random walk, is well studied recently.
In this model, each vertex $u$ sends tokens 
 one by one to neighboring vertices in the round robin fashion, 
i.e., $u$ serves tokens to a neighboring vertex $v$ with a ratio about $1/\delta(u)$, where $\delta(u)$ is the number of neighbors.

Yanovski et al.~\cite{YWB03} studied the asymptotic behavior of the rotor-router model, and proved that any rotor-router model always stabilizes to a traversal of an Eulerian cycle after $2mD$ steps at most, where $D$ denotes the diameter of the graph.
Bampas et al.~\cite{BGHI09} gave examples of which the stabilization time gets ${\rm \Omega }(mD)$. 
Their results imply that the cover time of a single token version of a rotor-router model is $\Theta (mD)$ in general.
Another approach to examine the cover time of the rotor-router model is connecting qualities of a random walk and the {\em visit frequency} $\X_v^{(T)}$ of the rotor-router model, 
where $\nv^{(T)}_v$ denotes the total number of times that tokens visited vertex $v$ by time $T$. 
Holroyd and Propp~\cite{HP10} showed that $|\pi_v-\nv^{(T)}_v/T|\leq K\pi_v/T$, where $K$ is an constant independent of $T$, and $\pi$ is the stationary distribution of the corresponding random walk. 
This theorem says that $\nv^{(T)}_v/T$ converges to $\pi_v$ as $T$ increasing. 
Using this fact, Friedrich and Sauerwald~\cite{FS10} gave upper bounds of the cover time for many classes of graphs. 

To speed up the cover time, the rotor-router model with $\tn >1$ tokens is studied by Dereniowski et al.~\cite{DKPU14}.
They gave an upper bound $\Order (mD/\log \tn)$
 for any graph when $\tn =\Order \bigl(\poly(n)\bigr)$ or $2^{\Order(D)}$, and also gave an example of ${\rm \Omega}(mD/\tn)$ as a lower bound. 
Kosowski and Pajak~\cite{KP14} gave a modified upper bound of the cover time for many graph classes by connecting $\nv^{(T)}_v$ and the corresponding simple random walk. 
They showed that the upper bound is $\Order \bigl(t^*+(\Delta/\delta)(mt^*/\tn) \bigr)$ for general graphs, where $\Delta/\delta$ is the maximum/minimum degree. 

Beyond the rotor-router model, which corresponds to a simple random walk, the {\em deterministic random walk} for general transition probabilities has been invented,
that is each vertex $u$ deterministically serves tokens on $u$ to a neighboring vertex $v$ with a ratio about $P_{u,v}$, where $P_{u,v}$ denotes the transition probability from $u$ to $v$ of a corresponding random walk (See Section~\ref{subsec:SRT} for the details).  
%
Holroyd and Propp~\cite{HP10} provides the {\em stack walk} (Shiraga et al.~\cite{SYKY13} called it {\em SRT-router model}), 
and showed a connection between the visit frequency and hitting probabilities. 
Shiraga et al.~\cite{SYKY13} investigated functional-router model, which is a more general framework, and gave an analysis on the {\em single vertex discrepancy} between the SRT-router model and its corresponding random walk. 
As far as we know, nothing is known about the cover time of deterministic random walks for general transition probabilities.
\paragraph{Our results}%
This paper is concerned with the cover time of the deterministic random walk according to general transition probabilities with $k$ tokens, while previous results studied the rotor-router model (corresponding to simple transition probabilities). 
We give an upper bound of the cover time for any SRT-router model imitating any ergodic and reversible transition matrix possibly containing irrational numbers (Theorem~\ref{thm:coverSRT}). 
Precisely, the upper bound is $\Order\bigl(t^*+m't^*/\tn \bigr)$ for any number of tokens $\tn \geq 1$, where $m'=\max_{u\in V}(\delta(u)/\pi_u)$. 
This is the first result of an upper bound of the cover time for deterministic random walks imitating general transition probabilities, as far as we know. 
Theorem~\ref{thm:coverSRT} implies that the upper bound of the cover time of the rotor-router model is $\Order\bigl(t^*+mt^*/\tn \bigr)$ for any graph (Corollary~\ref{cor:coverRR}). 
For $k=1$, this bound matches to the existing bound $\Order (mD)$ by \cite{YWB03} when $t^*=\Order (D)$.  
This bound is better than 
$\Order (mD/\log \tn)$ by \cite{DKPU14} when $t^*$ is small or $k$ is large. 
Our bound also improves the bound $\Order \bigl(t^*+(\Delta/\delta)(mt^*/\tn) \bigr)$ by \cite{KP14} in $\Delta/\delta$ factor for inhomogeneous graphs.

In our proof, we investigate the connection between the visit frequency $\X_v^{(T)}$ of the SRT-router model and the corresponding multiple random walks with general transition probabilities. 
This approach is an extension of \cite{HP10, FS10, KP14}. 
In precise, we show that $|\pi_v-(\nv^{(T)}_v/\tn T)|<K\pi_v/T$ holds for any reversible and ergodic transition matrices, where $\pi_v$ is the stationary distribution of the corresponding transition matrix and $K$ is constant independent of $T$. 
This upper bound extends the result of \cite{HP10} to $\tn>1$ tokens and general transition probabilities. 
\paragraph{Related topics for deterministic random walks}%
As a highly related topic, there are several results on the single vertex discrepancy between a configuration of tokens of a multiple deterministic random walk and an expected configuration of tokens of its corresponding random walk.
Rabani et al.~\cite{RSW98} gave an upper bound of the single vertex discrepancy of the {\em diffusive model}, and gave the framework of the analysis.
The single vertex discrepancy on several basic structures were widely studied,  
e.g., 
constant upper bound for the lattice~\cite{CS06, CDST07, DF09}, 
lower bound for the tree~\cite{CDFS10}, 
$d$-dimensional hyper cube~\cite{FGS12, AB13}, etc. 
Berenbrink et al.~\cite{BKKM15} gave a sophisticated upper bound on $d$-regular graphs.
To cope with general rational transition probabilities, rotor-router model on multidigraphs is studied in \cite{KKM12,KKM13}.
The SRT-router model is investigated in \cite{SYKY13, SYKY16}. They examined the discrepancy between this model and general Markov chains under natural assumptions. 
Recently, Chalopin et al.~\cite{CDGK15} gave the upper and lower bound of the stabilization time for the rotor-router model with many tokens.

\section{Preliminaries}\label{sec:model}%
\subsection{Random walk / Markov chain}%
Let $V=\{1,2,\ldots, n\}$ be a finite state set, and let $P\in \mathbb{R}^{n\times n}_{\geq 0}$ be a transition matrix on $V$. 
$P$ satisfies $\sum_{v\in V}P_{u,v}=1$ for any $u\in V$, where $P_{u,v}$ denotes the $(u,v)$-entry of $P$.
It is well known that any {\rm ergodic}\footnote{$P$ is ergodic if $P$ is {\rm irreducible} ($\forall u,v\in V, \exists t>0, P^t_{u,v}>0$) and {\rm aperiodic} ($\forall v\in V, {\rm GCD}\{t\in \mathbb{Z}_{>0}\mid P^t_{v,v}>0\}=1$).} $P$ has a unique {\em stationary distribution} $\pi\in \mathbb{R}^n_{>0}$ (i.e., $\pi P=\pi$), 
and the limit distribution is $\pi$ (i.e., $\lim_{t\to \infty}\xi P=\pi$ for any probability distribution $\xi$ on $V$). 
To discuss the {\em convergence} formally, we introduce the {\em total variation distance} and the {\em mixing time}.  
Let $\xi$ and $\zeta$ be probability distributions on $V$, then the total variation distance $\dtv$ between $\xi$ and $\zeta$ is defined by 
\begin{eqnarray}
\label{def:TV}
\dtv(\xi, \zeta) 
&\defeq &\max_{A\subseteq V} \left| \sum_{v\in A}(\xi_v-\zeta_v )\right|  
=\frac{1}{2} \left\|\xi-\zeta\right\|_1
=\frac{1}{2} \sum_{v\in V}|\xi_v-\zeta_v|. 
\end{eqnarray}
The {\em mixing time} of $P$ is defined by\footnote{$P^t_{u,v}$ denotes the $(u,v)$ entry of $P^t$, and $P^t_{u, \cdot}$ denotes the $u$-th row vector of $P^t$. } 
\begin{eqnarray}
\label{def:mix}
\tau(\varepsilon) \defeq 
\max_{u \in V} \min \left\{ t \in \mathbb{Z}_{\geq 0} \mid \dtv(P^t_{u, \cdot}, \pi) \leq \varepsilon \right\}
\end{eqnarray}
for $\varepsilon > 0$, and
\begin{eqnarray}
t^* \defeq \tau(1/4), 
\end{eqnarray}
which is often used as an important characterization of $P$ (cf.\cite{LPW08}). 

In this paper, we assume $P$ is ergodic and {\em reversible}. 
We call a $P$ is reversible if $\pi_uP_{u,v}=\pi_vP_{v,u}$ holds for any $u,v\in V$.
For example, transition matrices of the $\beta$-random walk~\cite{IKY09} and the Metropolis walk~\cite{NOSY10} are both reversible. 
\paragraph{Notations of multiple random walks}
Let $\mu^{(0)}=(\mu^{(0)}_1,\ldots,\mu^{(0)}_n)\in \mathbb{Z}^n_{\geq 0}$ denote an initial configuration of $\tn$ tokens over $V$. 
At each time step $t\in \mathbb{Z}_{\geq 0}$, each token on $v\in V$ moves independently to $u\in V$ with probability $P_{u,v}$. 
Let $\mu^{(t)}=(\mu^{(t)}_1,\ldots,\mu^{(t)}_n)\in \mathbb{R}^n_{\geq 0}$ denote the {\em expected} configuration of tokens at time $t\in \mathbb{Z}_{\geq 0}$: then $\mu^{(t)}=\mu^{(0)}P^t$ holds\footnote{In this paper, $(\mu^{(0)}P^t)_v$ denotes the $v$-th element of the vector $\mu^{(0)}P^t$, i.e., $(\mu^{(0)}P^t)_v=\sum_{u\in V}\mu^{(0)}_uP^t_{u,v}$. }. 
Note that the definitions of the mixing times say that $\dtv(\mu^{(t)}/\tn, \pi)\leq \varepsilon$ after $t\geq \tau (\varepsilon )$. 
%
\subsection{SRT-router model}\label{subsec:SRT}%
To imitate random walks with general transition probabilities possibly containing irrational numbers, 
the deterministic process based on {\em low-discrepancy sequences} (cf. \cite{AJJ10, T80}) were proposed, called {\em stack walk} in \cite{HP10} and {\em SRT-router model} in \cite{SYKY13}.
In this section, we describe the definition of this model. 

Let $\N(v)$ denote the (out-)neighborhood\footnote{If $P$ is reversible, $u\in \N(v)$ if and only if $v\in \N(u)$, and then we abuse $\N(v)$ for in-neighborhood of $v\in V$} of $v$, i.e., $\N(v)=\{u\in V\mid P_{v,u}>0\}$. 
In this model, $\tn$ tokens move according to {\em SRT-router} $\sigma_v:\mathbb{Z}_{\geq 0}\to \No(v)$ defined on each $v\in V$ for a given $P$. 
Given $\sigma_v(0),\ldots,\sigma_v(i-1)$, inductively $\sigma_v(i)$ is defined as follows. 
First, let
\begin{eqnarray*}
T_i(v)=\{u\in \No(v) \mid |\{ j\in [0,i) \mid \sigma_v(j) =u\}|  -(i+1)P_{v,u}<0\}, 
\end{eqnarray*}
where $[z,z')= \{z,z+1,\ldots,z'-1\}$ (and we remark $[z,z)=\emptyset$). 
Then, let $\sigma_v(i)$ be $u^*\in T_i(v)$ minimizing the value
\begin{eqnarray*}
\frac{\bigl|\{ j\in [0,i) \mid \sigma_v(j) =u\}\bigr| +1}{P_{v,u}}
\end{eqnarray*}
over choices $u\in T_i(v)$. If there are two or more such $u\in T_i(v)$, then let $u^*$ be the minimum in them in an arbitrary prescribed order. 
Then, the sequence $\sigma_v(0), \sigma_v(1), \ldots$ satisfies the following {\em low-discrepancy property} for any $v$ and $P$ (cf.~\cite{AJJ10, T80}). 
\begin{proposition}
\label{prop:SRT-disc}
\cite{AJJ10, T80}
For any $P$, 
\begin{eqnarray*}
\Bigl| \bigl| \{ j\in [0,z) \mid \sigma_v(j) =u\} \bigr| - z\cdotp P_{v,u} \Bigr| <1
\end{eqnarray*}
holds for any $v,u\in V$ and for any integer $z>0$. 
\end{proposition}

Let $\chi^{(0)}=\mu^{(0)}$ and $\chi^{(t)}\in \mathbb{Z}^n_{\geq 0}$ denote the configuration of $\tn$ tokens at time $t\in \mathbb{Z}_{\geq 0}$ in a SRT-router model ($\sum_{v\in V}\chi^{(t)}_v=\tn$). 
Then, SRT-router model works as follows.
At first time step ($t=0$), there are $\chi_v^{(0)}$ tokens on vertex $v$, and each $v$ serves tokens to neighbors according to $\sigma_v(0),\sigma_v(1),\ldots,\sigma_v(\chi_v^{(0)}-1)$. 
In other words, $|\{j\in [0,\chi_v^{(0)})\mid \sigma_v(j)=u\}|$ tokens move from $v$ to $u$, and $\chi^{(1)}_u=\sum_{v\in V}|\{j\in [0,\chi_v^{(0)})\mid \sigma_v(j)=u\}|$. 
Next time step ($t=1$), there are $\chi_v^{(1)}$ tokens on vertex $v$, and each $v$ serves tokens to neighbors according to $\sigma_v(\chi_v^{(0)}),\sigma_v(\chi_v^{(0)}+1),\ldots,\sigma_v(\chi_v^{(0)}+\chi_v^{(1)}-1)$, and $\chi^{(2)}$ is defined in a similar way.
In general, let $Z_{v,u}^{(t)}$ denote the number of tokens moving from $v$ to $u$ at time $t$. Then $Z_{v,u}^{(t)}$ is defined as 
\begin{eqnarray}
\label{def:Z}
Z_{v,u}^{(t)}=\left|\left\{ j\in [0,\chi_v^{(t)}) \mid \sigma_v(\X_v^{(t)}+j)=u \right\}\right|, 
\end{eqnarray}
where
$
\X^{(T)}=\sum_{t=0}^{T-1}\chi^{(t)}
$ (and we remark $\X^{(0)}_v=0$ for any $v\in V$), and $\chi^{(t+1)}$ is defined by 
\begin{eqnarray}
\label{eq:inZ}
\chi_u^{(t+1)}=\sum_{v\in V}Z_{v,u}^{(t)}=\sum_{v\in \Ni(u)}Z_{v,u}^{(t)}
\end{eqnarray}
for any $u\in V$. 
Note that 
\begin{eqnarray}
\label{eq:outZ}
\sum_{u\in V}Z_{v,u}^{(t)}=\sum_{u\in \No(v)}Z_{v,u}^{(t)}=\chi_v^{(t)}
\end{eqnarray}
holds for any $v\in V$. 
For the SRT-router model, we have the following basic proposition, based on Proposition~\ref{prop:SRT-disc}. 
\begin{proposition}
\label{obs:discZ}
\begin{eqnarray*}
\left|\sum_{t=0}^{T}(Z_{v,u}^{(t)}-\chi_v^{(t)}P_{v,u})\right| <1
\end{eqnarray*}
holds for any $P$ and for any $T\geq 0$. 
\end{proposition}
\begin{proof}
From the definition of $Z_{v,u}^{(t)}$, it is not difficult to check that
\begin{eqnarray*}
\sum_{t=0}^{T}Z_{v,u}^{(t)}
&=&\sum_{t=0}^{T}\left|\left\{ j\in [0,\chi_v^{(t)}) \mid \sigma_v(\X_v^{(t)}+j)=u \right\}\right|\\
&=&\sum_{t=0}^{T}\left|\left\{ j\in [\X_v^{(t)},\X_v^{(t)}+\chi_v^{(t)}) \mid \sigma_v(j)=u \right\}\right|
=\left|\left\{ j\in [0,\X_v^{(T+1)}) \mid \sigma_v(j)=u \right\}\right|
\end{eqnarray*}
and 
$
\sum_{t=0}^{T}\chi_v^{(t)}P_{v,u}=\X_v^{(T+1)}P_{v,u}.
$
Then, Proposition~\ref{obs:discZ} is obtained by Proposition~\ref{prop:SRT-disc} by letting $z=\X_v^{(T+1)}$.
\end{proof}
\section{Analysis of the Visit Frequency}\label{sec:visit}%
%
As a preliminary of the analysis of the cover time of the SRT-router model, we investigate the upper bound of $|\X^{(T)}_w-\M^{(T)}_w|$, where $\M^{(T)}=\sum_{t=0}^{T-1}\mu^{(t)}$ (and we remark that $\M^{(0)}_v=0$ for any $v\in V$). 
Let $\delta(v)=|\N(v)|$ and $\Delta=\max_{v\in V}\delta(v)$. 
\begin{theorem}
\label{thm:visit}
Suppose that $P$ is ergodic and reversible. Then, 
\begin{eqnarray*}
\left|\X_w^{(T)}-\M^{(T)}_w\right|\leq 3\pi_w t^* \max_{u\in V}\frac{\delta(u)}{\pi_u}=\Order\left(\frac{\pi_{\max}}{\pi_{\min}}t^*\Delta\right)
\end{eqnarray*}
holds for any $w\in V$ and for any $T>0$.
\end{theorem}
From Theorem~\ref{thm:visit}, we get the following corollary~\ref{cor:visitpi}, like Theorem 4 of~\cite{HP10}. 
\begin{corollary}
\label{cor:visitpi}
Suppose that $P$ is ergodic and reversible. Then, 
\begin{eqnarray*}
\left|\pi_w-\frac{\X_w^{(T)}}{\tn T}\right|\leq \frac{3t^*}{2T}+\frac{3\pi_w t^* \max_{u\in V}\frac{\delta(u)}{\pi_u}}{\tn T}=\frac{K\pi_w}{T}
\end{eqnarray*}
holds for any $w\in V$ and for any $T>0$, where $K=\Order(\frac{t^*}{\pi_w}+\frac{t^*\Delta}{\pi_{\min}\tn })$ is a constant independent of $T$.
\end{corollary}
Note that Corollary~\ref{cor:visitpi} gives the upper bound for SRT-router models with $\tn$ tokens, 
while Theorem 4 of~\cite{HP10} is for rotor-router models with a single token. 
Corollary~\ref{cor:visitpi} also means that $\left|\pi_w-\frac{\X_w^{(T)}}{\tn T}\right|\leq \varepsilon$ if
$
T\geq 3\left( \frac{1}{2}+ \frac{\pi_w\Delta}{\pi_{\min}\tn}\right)t^*\varepsilon^{-1}
$.

\vspace{1em}
To prove the Theorem~\ref{thm:visit}, we begin with the following lemma.
In the following arguments, we assume that $P$ is ergodic and reversible.
\begin{lemma}
\label{lem:visitlemma1}
\begin{eqnarray*}
\X^{(T)}_w-\M^{(T)}_w=\sum_{t=0}^{T-2}\sum_{u\in V}\sum_{v\in \N(u)}\sum_{\ti=0}^{T-t-2}(Z_{v,u}^{(\ti)}-\chi_v^{(\ti)}P_{v,u})(P_{u,w}^{t}-\pi_w)
\end{eqnarray*}
holds for any $w\in V$ and for any $T>1$. 
\end{lemma}
\begin{proof}
We use the following lemma to prove Lemma~\ref{lem:visitlemma1}. 
\begin{lemma}
\cite{SYKY13}~(Lemma 4.1.)
\label{lem:discdetrw}
\begin{eqnarray*}
\chi_w^{(T)}-\mu_w^{(T)}=\sum_{t=0}^{T-1}\sum_{u\in V}\sum_{v\in \N(u)}(Z_{v,u}^{(t)}-\chi_v^{(t)}P_{v,u})(P_{u,w}^{T-t-1}-\pi_w)
\end{eqnarray*}
holds for any $w\in V$ and for any $T>0$.
\shortqed
\end{lemma}
By the definitions of $\X^{(T)}, \M^{(T)}$ and Lemma~\ref{lem:discdetrw}, 
\begin{eqnarray}
\X^{(T)}_w-\M^{(T)}_w
&=&\sum_{\ti=0}^{T-1}(\chi^{(\ti)}-\mu^{(\ti)})
=\sum_{\ti=1}^{T-1}(\chi^{(\ti)}-\mu^{(\ti)}) \nonumber \\
&=&\sum_{\ti=1}^{T-1}\sum_{t=0}^{\ti-1}\sum_{u\in V}\sum_{v\in \N(u)}(Z_{v,u}^{(t)}-\chi_v^{(t)}P_{v,u})(P_{u,w}^{\ti-t-1}-\pi_w)
\label{eq:visitl1}
\end{eqnarray}
holds. The second equation holds since $\chi^{(0)}=\mu^{(0)}$.
Let $\phi_u^{(t)}=\sum_{v\in \N(u)}(Z_{v,u}^{(t)}-\chi_v^{(t)}P_{v,u})$, for convenience. Then, 
\begin{eqnarray}
\eqref{eq:visitl1}&=&\sum_{\ti=1}^{T-1}\sum_{t=0}^{\ti-1}\sum_{u\in V}\phi^{(t)}_u(P_{u,w}^{\ti-t-1}-\pi_w)
=\sum_{u\in V}\sum_{\ti=1}^{T-1}\sum_{t=0}^{\ti-1}\phi^{(\ti-t-1)}_u(P_{u,w}^{t}-\pi_w)
\label{eq:visitl2}
\end{eqnarray}
holds.
Carefully exchanging the variables of the summation, we obtain
\begin{eqnarray}
\sum_{\ti=1}^{T-1}\sum_{t=0}^{\ti-1}\phi^{(\ti-t-1)}_u(P_{u,w}^{t}-\pi_w)
=\sum_{t=0}^{T-2}\sum_{\ti=t+1}^{T-1}\phi^{(\ti-t-1)}_u(P_{u,w}^{t}-\pi_w)
=\sum_{t=0}^{T-2}\sum_{\ti=0}^{T-t-2}\phi^{(\ti)}_u(P_{u,w}^{t}-\pi_w). 
\label{eq:visitl3}
\end{eqnarray}
Combining \eqref{eq:visitl2} and \eqref{eq:visitl3}, we obtain
\begin{eqnarray*}
\eqref{eq:visitl2}
&=&\sum_{u\in V}\sum_{t=0}^{T-2}\sum_{\ti=0}^{T-t-2}\phi^{(\ti)}_u(P_{u,w}^{t}-\pi_w)\\
&=&\sum_{u\in V}\sum_{t=0}^{T-2} \sum_{\ti=0}^{T-t-2}\sum_{v\in \N(u)} (Z_{v,u}^{(\ti)}-\chi_v^{(\ti)}P_{v,u})(P_{u,w}^{t}-\pi_w). 
\end{eqnarray*}
\end{proof}
\begin{proof}[of Theorem~\ref{thm:visit}]
It is trivial for $T=1$, hence we assume $T>1$. 
By Lemma~\ref{lem:visitlemma1} and Proposition~\ref{obs:discZ}, 
\begin{eqnarray}
\left| \X_w^{(T)}-\M_w^{(T)}\right|
&=& \left| \sum_{t=0}^{T-2}\sum_{u\in V}\sum_{v\in \N(u)}\sum_{\ti=0}^{T-t-2}(Z_{v,u}^{(\ti)}-\chi_v^{(\ti)}P_{v,u})(P_{u,w}^{t}-\pi_w)\right|\nonumber \\
&\leq &\sum_{t=0}^{T-2}\sum_{u\in V}\sum_{v\in \N(u)}\left| \sum_{\ti=0}^{T-t-2}(Z_{v,u}^{(\ti)}-\chi_v^{(\ti)}P_{v,u}) \right| \left|P_{u,w}^{t}-\pi_w\right|\nonumber \\
&<&\sum_{t=0}^{T-2}\sum_{u\in V}\sum_{v\in \N(u)} \left|P_{u,w}^{t}-\pi_w\right|
=\sum_{t=0}^{T-2}\sum_{u\in V}\delta(u)\left|P_{u,w}^{t}-\pi_w\right|
\label{eq:visitt1}
\end{eqnarray}
holds. 
By the reversibility of $P$, 
\begin{eqnarray}
\eqref{eq:visitt1}
&=&\sum_{t=0}^{T-2}\sum_{u\in V}\delta(u)\left|\frac{\pi_w}{\pi_u}(P_{w,u}^{t}-\pi_u)\right|
\leq \pi_w \max_{u\in V} \frac{\delta(u)}{\pi_u} \sum_{t=0}^{T-2}\sum_{u\in V}\left|P_{w,u}^{t}-\pi_u \right|
\label{eq:visitt2}
\end{eqnarray}
holds. By the definition of total variation distance~\eqref{def:TV}, 
\begin{eqnarray}
\label{eq:visitt3}
\sum_{t=0}^{T-2}\sum_{u\in V}\left|P_{w,u}^{t}-\pi_u \right|
=2\sum_{t=0}^{T-2}\dtv(P^t_{w,\cdot},\pi)
\end{eqnarray}
holds. Now, we use the following lemma. 
\begin{lemma}
\cite{SYKY13} (Lemma 4.2.)
\label{lemm:dtsum}
For any $v\in V$ and for any $T>0$, 
\begin{eqnarray*}
 \sum_{t=0}^{T-1} \dtv\left( P^t_{v, \cdot}, \pi \right) 
 \leq \frac{1-\gamma }{1-2\gamma }\, \tau(\gamma )
\end{eqnarray*}
holds for any $\gamma$ $(0<\gamma <1/2)$. 
\end{lemma} 
Thus, we have
\begin{eqnarray}
\label{eq:visitt4}
\eqref{eq:visitt3}
\leq 2\cdotp \frac{1-(1/4)}{1-2\cdotp(1/4)}\tau(1/4)=3t^*
\end{eqnarray}
and we obtain the claim. 
\end{proof}
\begin{proof}[of Corollary~\ref{cor:visitpi}]
Notice that
\begin{eqnarray}
\left|\pi_w-\frac{\X_w^{(T)}}{\tn T}\right|
&=&\frac{\left|\tn T\pi_w-\X_w^{(T)}\right|}{\tn T}
\leq \frac{\left|\tn T\pi_w-\M_w^{(T)}\right| + \left|\M_w^{(T)}-\X_w^{(T)}\right|}{\tn T}\nonumber \\
&\leq & \frac{\left| \M_w^{(T)}-\tn T\pi_w\right|}{\tn T}+\frac{3\pi_w t^* \max_{u\in V}\frac{\delta(u)}{\pi_u}}{\tn T}, \nonumber
\end{eqnarray}
where the last inequality follows Theorem~\ref{thm:visit}.
Thus, it is sufficient to prove that $|\M_w^{(T)}-\tn T\pi_w|\leq 3\tn t^*/2$. 
Note that
$
\sum_{t=0}^{T-1}\sum_{u\in V}\mu^{(0)}\pi_w=\tn T\pi_w
$
holds since $\sum_{v\in V}\mu^{(0)}=\tn$ from the definition,  
and also note that
$
\M_w^{(T)}=\sum_{t=0}^{T-1}\mu^{(t)}_w=\sum_{t=0}^{T-1}\sum_{u\in V}\mu^{(0)}_uP_{u,w}^t
$
holds by the definitions.
Then, 
\begin{eqnarray}
\left| \M_w^{(T)}-\tn T\pi_w \right|
&=&\left| \sum_{t=0}^{T-1}\sum_{u\in V}\mu^{(0)}_uP^t_{u,w}- \sum_{t=0}^{T-1}\sum_{u\in V}\mu^{(0)}\pi_w \right|
=\left| \sum_{t=0}^{T-1}\sum_{u\in V}\mu^{(0)}_u(P^t_{u,w}-\pi_w) \right|\nonumber \\
&\leq &\sum_{u\in V}\mu^{(0)}_u\sum_{t=0}^{T-1}|P^t_{u,w}-\pi_w|
\label{eq:visitc2}
\end{eqnarray}
holds. 
By Lemma~\ref{lemm:dtsum} and the definition of total variation distance~\eqref{def:TV}, 
\begin{eqnarray}
\sum_{t=0}^{T-1}|P^t_{u,w}-\pi_w|
\leq \sum_{t=0}^{T-1}\dtv(P^t_{u,\cdot},\pi)
\leq \frac{3}{2}t^*. 
\label{eq:visitc3}
\end{eqnarray}
Combining \eqref{eq:visitc2} and \eqref{eq:visitc3}, $|\M_w^{(T)}-\tn T\pi_w|\leq 3\tn t^*/2$ holds, and we obtain the claim. 
\end{proof}
%
\section{Bound of the Cover Time}\label{sec:cover}%
%
Combining techniques of the analysis of the visit frequency and reversible Markov chains, we obtain the cover time of SRT-router models.
Let
\begin{eqnarray}
\CT=\min\left\{T\in \mathbb{Z}_{\geq 0}\mid \X_v^{(T)} \geq 1\ holds\ for\ any\ v\in V\right\}. 
\end{eqnarray}
First, we show the following theorem. 
\begin{theorem}
\label{thm:coverSRT}
Suppose $P$ be ergodic and reversible. Then, 
\begin{eqnarray*}
\CT\leq 2t^*+1+\frac{12\max_{u\in V}\frac{\delta(u)}{\pi_u}\cdotp t^*}{\tn}
=\Order\left(\max\left\{\frac{t^*\Delta}{\pi_{\min}\tn},t^*\right\}\right)
\end{eqnarray*}
holds for any initial configuration of $\tn \geq1$ tokens.
\end{theorem}
Theorem~\ref{thm:coverSRT} is the first result of the cover time for deterministic random walks imitating general transition probabilities possibly containing irrational transition probabilities. 
Applying Theorem~\ref{thm:coverSRT} to the transition matrix of simple random walk on $G$, we obtain the following corollary.
\begin{corollary}
\label{cor:coverRR}
For any $G$ and for any initial configuration of $\tn \geq1$ tokens, 
\begin{eqnarray*}
\CT\leq 2t^* + 1+\frac{24mt^*}{\tn}=\Order\left(\max\left\{\frac{mt^*}{\tn},t^*\right\}\right)
\end{eqnarray*}
holds for any rotor-router model on $G$, where $t^*$ is the mixing time of the simple random walk on $G$. 
\end{corollary}
The upper bound of \cite{KP14} (Theorem 4.1, proposition 4.2, and Theorem 4.5) is $\Order \bigl( t^* +(\Delta/\delta)(mt^*/\tn)\bigr)$, where $\Delta/\delta$ is the maximum/minimum degree of the graph.
Hence Corollary~\ref{cor:coverRR} improves this bound for inhomogeneous graphs.
Compare to the $\Order (mD/\log k)$ by \cite{DKPU14} (Theorem 3.3 and 3.7), our bound is better when $t^*=\Order\bigl( D(k/\log k) \bigr)$ (when $t^*$ is small or $\tn$ is large). 
%
%
%

\vspace{1em}
To prove Theorem~\ref{thm:coverSRT}, we check the following lemma. 
\begin{lemma}
\label{obs:separation}
Suppose $P$ is ergodic and reversible. Then,
\begin{eqnarray*}
P^t_{u,w}\geq \frac{\pi_w}{4}
\end{eqnarray*}
folds for any $u,w\in V$ if $t\geq 2t^*$.
\end{lemma}
\begin{proof}
The {\em separation distance}~\cite{AD87} is defined by 
\begin{eqnarray}
s(t)=\max_{u,v\in V}\left(1-\frac{P^t_{u,v}}{\pi_v}\right).
\end{eqnarray}
This distance satisfies $s(t+t')\leq s(t)s(t')$ for any $t, t'\geq 1$ (submultiplicativity property, Lemma 3.7 of \cite{AD87}). We have the following lemma for the reversible $P$. 
\begin{lemma}\cite{LPW08} (Lemma 19.3.)
\label{lem:sep}
Suppose $P$ is reversible. then, 
\begin{eqnarray*}
s(2t)\leq 1-\bigl(1-\bar{d}(t)\bigr)^2
\end{eqnarray*}
holds for any $t\geq 0$, where
$
\bar{d}(t)=\max_{u,v\in V}\dtv(P^t_{u,\cdot},P^t_{v,\cdot})
$. 
\end{lemma} 
It is known that 
\begin{eqnarray}
\bar{d}(t^*)\leq \frac{1}{2}
\end{eqnarray}
holds when $P$ is ergodic (see (4.34) of \cite{LPW08}). 
Combining these facts, we have
\begin{eqnarray*}
1-\frac{P^t_{u,w}}{\pi_w}
&\leq &s(t)
\leq s(2t^*) 
\leq 1-(1-\bar{d}(t^*))^2
\leq 1-\left(1-\frac{1}{2}\right)^2
=\frac{3}{4}, 
\end{eqnarray*}
and we obtain the claim.
\end{proof}
%

\begin{proof}[of Theorem~\ref{thm:coverSRT}]
Lemma~\ref{obs:separation} gives us a lower bound of $P^t_{u,w}$ for any $u,w\in V$, $t\geq 2t^*$ and for any reversible and ergodic $P$. It provides a lower bound of $\M_w^{(T)}$, like \cite{KP14}. 
\begin{eqnarray}
\M_w^{(T)}
&=&\sum_{t=0}^{T-1}\sum_{u\in V}\mu^{(0)}_uP^t_{u,w}
\geq \sum_{t=2t^*}^{T-1}\sum_{u\in V}\mu^{(0)}_uP^t_{u,w}
\geq \sum_{t=2t^*}^{T-1}\sum_{u\in V}\mu^{(0)}_u\frac{\pi_w}{4}
=\frac{\tn \pi_w(T-2t^*)}{4}. 
\label{eq:coverSRT1}
\end{eqnarray}
By Theorem~\ref{thm:visit} and \eqref{eq:coverSRT1}, we obtain that
\begin{eqnarray}
\X_w^{(T)}
&\geq &\M_w^{(T)}-3\pi_w t^* \max_{u\in V}\frac{\delta(u)}{\pi_u}
\geq \frac{\tn \pi_w(T-2t^*)}{4}-3\pi_w t^* \max_{u\in V}\frac{\delta(u)}{\pi_u}. 
\label{eq:thmsrt2}
\end{eqnarray}
Notice that \eqref{eq:thmsrt2} implies
\begin{eqnarray*}
\X_w^{(T')}>0
\label{eq:ct0}
\end{eqnarray*}
for any $w\in V$ and for {\em any} $T'\in \mathbb{Z}_{\geq 0}$ satisfying
\begin{eqnarray*}
T'>2t^* + \frac{12t^* \max_{u\in V}\frac{\delta(u)}{\pi_u}}{\tn}. 
\end{eqnarray*}
The fact \eqref{eq:ct0} implies that $\CT \leq T'$, and we obtain the claim.
\end{proof}

\begin{proof}[of Corollary~\ref{cor:coverRR}]
Note that a SRT-router model corresponding to a simple random walk on $G$ is exactly a rotor-router model on $G$, and we see that $\max_{u\in V}\frac{\delta(u)}{\pi_u}=2m$, since $\pi_u=\frac{\delta(u)}{2m}$. Thus, 
\begin{eqnarray*}
\CT\leq 2t^* + 1+\frac{24mt^*}{\tn}
\end{eqnarray*}
holds by Theorem~\ref{thm:coverSRT}. 
\end{proof}
%
\section{Concluding Remarks}%
In this paper, we gave techniques to examine the visit frequency $X_v^{(T)}$ of the SRT-router model with $\tn>1$ tokens, and gave an upper bound of the cover time for any ergodic and reversible $P$. 
Also, our upper bound improve the upper bound of the previous results of the rotor-router model with $\tn >1$ tokens in many cases. 
A better upper bound of the cover time by derandomizing a specific {\em fast} random walk (e.g., $\beta$-random walk, Metropolis walk) is a challenge.

\section*{Acknowledgements}%
The author would like to thank Prof. Kijima for his comments on the manuscript.
The author is also grateful to Prof. Sauerwald and Dr. Pajak for a
discussion on the topic.
This work is supported by JSPS KAKENHI Grant Number 15J03840.
The author also gratefully acknowledge to the ELC project (Grant-in-Aid for
Scientific Research on Innovative Areas MEXT Japan) for encouraging
the research presented in this paper.

\bibliographystyle{abbrv}

\end{document}